\documentclass[runningheads]{llncs}
\usepackage[T1]{fontenc}

\usepackage[margin=1in]{geometry}
\usepackage{lmodern}
\usepackage[T1]{fontenc}
\usepackage[utf8]{inputenc}
\usepackage[tbtags]{amsmath}
\allowdisplaybreaks
\usepackage{bbm}
\usepackage{amssymb}
\usepackage[hidelinks]{hyperref}
\usepackage[numbers]{natbib}
\usepackage{doi}
\usepackage{tikz} 
\usepackage{pgf,pgfplots}
\pgfplotsset{compat=1.14}
\usepgfplotslibrary{fillbetween}
\usepackage{hyphenat}

\usepackage[capitalise,nameinlink]{cleveref}

\usepackage{dsfont}
\usepackage{microtype,bm}
\usepackage{booktabs}
\usepackage{colortbl}
\usepackage{xcolor}
\usepackage{setspace}
\usepackage{nicefrac}
\usepackage{graphicx}

\usepackage{thmtools} 
\usepackage{thm-restate}

\usepackage{algorithm}
\usepackage[noend]{algpseudocode}
\algblockdefx[IF]{If}{EndIf}[1]{\textbf{If} #1\textbf{, then}}{}
\algtext*{EndIf}
\algblockdefx[FOR]{For}{EndFor}[1]{\textbf{For} #1\textbf{, do}}{}
\algtext*{EndFor}
\algblockdefx[WHILE]{While}{EndWhile}[1]{\textbf{While} #1\textbf{, do}}{}
\algtext*{EndWhile}
\algblockdefx[FUNCTION]{Function}{EndFunction}[2]{\textbf{Function} \textsc{#1}(#2)}{}
\algtext*{EndFunction}
\algdef{SE}[SUBALG]{Indent}{EndIndent}[1]{\textbf{#1:}}{}
\algtext*{EndIndent}

\usepackage{tikz,amsmath}
\usetikzlibrary{calc}

\usepackage[labelformat=simple]{subcaption}

\usepackage{microtype,bm,xspace}

\usepackage{multirow}

\usepackage{color}

\newcommand{\M}{\mathcal{M}}
\newcommand{\E}[1]{\mathbb{E}\left[#1\right]}
\newcommand{\ind}[1]{\mathbb{I}_{#1}}
\DeclareMathOperator*{\argmax}{arg\,max}

\newcommand{\Bid}[1]{\mathcal B_{#1}}
\newcommand{\bid}[2]{b_{#1\ifx&#2&\else,\fi#2}}
\newcommand{\payment}[1]{p_{#1}}
\newcommand{\critprice}[2]{\pi_{#1\ifx&#2&\else,\fi#2}}
\newcommand{\crititem}[1]{j(#1)}

\newcommand{\sw}{\textsc{SW}}
\newcommand{\opt}{\textsc{OPT}}

\newcommand{\ptgreedy}{\textsc{HonestPerturbedGreedy}\xspace}
\newcommand{\tgreedy}{\textsc{HonestGreedy}\xspace}
\newcommand{\pgreedy}{\textsc{Perturbed-Greedy}\xspace}
\newcommand{\greedy}{\textsc{Greedy}\xspace}
\newcommand{\ranking}{\textsc{Ranking}\xspace}

\newcommand{\expexp}{\textsc{Explore-Exploit Mechanism}\xspace}

\begin{document}
\title{Truthful Matching with Online Items and Offline Agents}
\titlerunning{Truthful Matching with Online Items and Offline Agents}
%
\author{Michal Feldman\inst{1,2}\and
Federico Fusco\inst{3}\and
Stefano Leonardi\inst{3}\and
Simon Mauras\inst{1}\and
Rebecca Reiffenh\"auser\inst{3}}
\authorrunning{M. Feldman, F. Fusco, S. Leonardi, S. Mauras and R. Reiffenh\"auser}
%
\institute{Blavatnik School of Computer Science, Tel Aviv University, Israel \\
\email{\{mfeldman,smauras\}@tauex.tau.ac.il}
\and
    Microsoft Research Israel
\and
Department of Computer, Control
and Management Engineering ``Antonio Ruberti'', \\
Sapienza University of Rome, Italy \\
\email{\{fuscof,leonardi,rebeccar\}@diag.uniroma1.it}}
\maketitle              
\begin{abstract}
    We study truthful mechanisms for welfare maximization in online bipartite matching. 
    In our (multi-parameter) setting, every buyer is associated with a (possibly private) desired set of items, and has a private value for being assigned an item in her desired set. 
    Unlike most online matching settings, where agents arrive online, in our setting the items arrive online in an adversarial order while the buyers are present for the entire duration of the process. 
    This poses a significant challenge to the design of truthful mechanisms, due to the ability of buyers to strategize over future rounds. 
    We provide an almost full picture of the competitive ratios in different scenarios, including myopic vs. non-myopic agents, tardy vs. prompt payments, and private vs. public desired sets.
    Among other results, we identify the frontier for which the celebrated $e/(e-1)$ competitive ratio for the vertex-weighted online matching of Karp, Vazirani and Vazirani extends to truthful agents and online items.
\end{abstract}
\section{Introduction}
Matching in bipartite graphs is a fundamental model that has found numerous applications with the growth of the Internet. 
Some examples include items and buyers in e-commerce, drivers and passengers in ride-sharing platforms, ad slots and advertisers in online ad auctions, and jobs and workers in online labor markets.
In these applications, it is common that vertices on one side are known from the outset, while vertices from the other side arrive one-by-one in an online fashion. 
Upon the arrival of an online vertex, its information is revealed (containing, e.g., its set of adjacent edges, and their weights), and the algorithm has to immediately and irrevocably decide either to match it with an available offline partner or leave it unmatched forever. The goal is to maximize the sum of the weights along the matched edges.

A celebrated result in online matching by Karp, Vazirani, and Vazirani \cite{karp1990optimal} shows that in the unweighted setting, a simple randomized strategy, called \ranking, achieves a competitive ratio of $e/(e-1)$, and this is optimal. This result extends to the setting where the vertices on the offline side are weighted and the objective is to maximize the sum of the weights of the matched vertices. Although the original algorithm for this problem, \pgreedy \cite{AggarwalGKM11}, was designed for non-strategic settings, online matching problems have also been studied in the presence of {\em strategic} agents \citep[e.g.,][]{KrystaV12,reiffenhauser2019optimal,DuttingFLLR21}. This is not a mere theoretical exercise: online matching is used in many situations where the parties involved are interested in misreporting their true valuations to obtain a better outcome: e.g., combinatorial and ad-auctions, kidney exchange, school-student matching, and house allocation. In the presence of strategic agents, an agent's value is her private information, and is not directly available to the mechanism designer.
The main challenge here is to design \emph{incentive-compatible} or \emph{truthful} mechanisms which, besides finding a good matching, also ensure that it is in the agents' best interest to report their true values.
In addition to making decisions regarding the matching itself, such mechanisms can also charge some payment from the agents in order to incentivize them to truthfully report their values. Here, each agent strives to maximize her \emph{quasi-linear utility}, which is defined as the value she obtains from her assigned item, minus the payment she has to make.

In almost all previous studies, the agents are represented by the vertices in the online side, while the items they are competing over are available offline. 
In many natural Internet applications, e.g., selling advertising opportunities via repeated auctions, the agents are fixed and observe a stream of items arriving online.
This motivates the study of a \emph{reversed} online matching problem, where the offline side is strategic on her value and on her set of desired items that arrive online.
This particular variant has been considered thus far only in very restricted settings \cite{cole2008prompt, DPZ21}.
This is not a coincidence: when agents are present throughout the entire matching process, many new manipulation opportunities arise, and incentivizing truthful behavior is significantly more challenging. Indeed, the online nature of the problem forces any mechanism to repeatedly make irrevocable decisions upon the arrival of goods, lacking knowledge about future opportunities that might arise to the participating agents.
The agents --- possibly aware of those future opportunities --- may strategize to gain benefits in the future, challenging standard tools that have been applied in cases where agents arrive online.

Our work provides a systematic analysis of this scenario, and gives (almost) tight competitive ratios under a rich combination of natural assumptions. We study this problem along different dimensions, as follows. First, we consider two types of agents --- myopic and non-myopic --- that are characterized by the different information they have on the instance. Myopic agents make strategic considerations that are limited to the current time step, without looking forward into the future (see, e.g., Deng,~Panigrahi~and~Zhang~\cite{DPZ21}), whereas non-myopic agents optimize across multiple time steps, using the up-front knowledge of the underlying (online) graph.
The assumption of myopic agents clearly eradicates some of the difficulties of designing (almost tight) online mechanisms with offline strategic agents thus allowing to derive efficient mechanisms from known online matching algorithms, e.g., from Aggarwal~et~al.~\cite{AggarwalGKM11}.
Second, we consider two types of private information. In the first scenario we consider, an agent's private information consists of her private value for her desired items, but the set of desired items is publicly known. In the second scenario, both the value and the set of desired items are private information.

Notably, in both cases the graph structure is revealed to the mechanism step-by-step, upon the arrival of every item. Finally, we distinguish between \emph{prompt} and \emph{tardy} mechanisms.
Both types of mechanisms make allocation decisions immediately. 
However, they differ in the time at which they make payment decisions. 
Prompt mechanisms make payment decisions  immediately upon allocation, while tardy mechanism may delay payment decisions to the end of the entire process. 

\subsection{Our Results and Techniques}

We conduct a systematic study of online bipartite matching with online items and offline agents, in a variety of scenarios and we provide (almost) tight bounds for the settings of interest, as summarized in Table \ref{table:results}. \vspace{-7pt}

\paragraph{Myopic agents.}
We start by investigating the simpler setting of myopic agents. 
These agents care only about their instantaneous utility, and do not strategize over the future. As such, we only consider prompt mechanisms for this type of agents.
Exploiting the myopic nature of the agents, it is not difficult to turn the best (non-truthful) algorithms into (truthful) mechanisms. In particular, we construct a deterministic prompt mechanism based on the greedy matching algorithm that is guaranteed to achieve at least a half of the optimal welfare. We also give a randomized prompt mechanism based on the algorithm for weighted online matching \cite{AggarwalGKM11}, which is $e/(e-1)$-competitive.
This shows that the transition from non-strategic agents to strategic myopic agents does not
lead to a deterioration in efficiency guarantees. Notably, for the special case we study, our bounds for myopic online matching improve vastly over those obtained by Deng,~Panigrahi~and~Zhang~\cite{DPZ21} for general XOS valuations. The results for myopic agents are presented in \Cref{table:myopic}.\vspace{-7pt}

\begin{table}[t]
    \footnotesize
    \captionsetup[subtable]{justification=centering}
    \begin{subtable}{.32\linewidth}
        \begin{tabular}{|c|c|c|}
        \cline{2-3}
        \multicolumn{1}{c|}{}&
        \multirow{2}{*}{deterministic}&
        \multirow{2}{*}{randomized}\\
        \multicolumn{1}{c|}{}&
        &\\
    \hline
        \multirow{2}{*}{prompt}&
        $2$&
        $e/(e-1)$\\
        &
        (\Cref{thm:myopic-greedy})&
        (\Cref{thm:myopic-rank})\\
        \hline
\end{tabular}
        \caption{Myopic agents\\~}
        \label{table:myopic}
    \end{subtable}
    \begin{subtable}{.36\linewidth}
        \begin{tabular}{|c|c|c|}
        \cline{2-3}
        \multicolumn{1}{c|}{}&
        \multirow{2}{*}{deterministic}&
        \multirow{2}{*}{randomized}\\
        \multicolumn{1}{c|}{}&
        &\\
    \hline
        \multirow{2}{*}{tardy}&
        $2$&
        $e/(e-1)$\\
        &
        (\Cref{thm:public-greedy})&
        (\Cref{thm:public-greedy})\\
    \hline
        \multirow{2}{*}{prompt}&
        $\geq \nu$&
        $\Omega(\log\nu/\log\log\nu)$
        \\
        &
        (\Cref{thm:expost-det-lb})&
        (\Cref{thm:expost-rand-lb})
        \\
    \hline
\end{tabular}
        \caption{Non-myopic agents\\with \emph{public} graph edges}
        \label{table:public}
    \end{subtable} 
    \begin{subtable}{.31\linewidth}
        \begin{tabular}{|c|c|c|}\cline{2-3}
        \multicolumn{1}{c|}{}&
        \multirow{2}{*}{deterministic}&
        \multirow{2}{*}{randomized}\\
        \multicolumn{1}{c|}{}&
        &\\
    \hline
        \multirow{2}{*}{tardy}&
        \multirow{2}{*}{\tiny equiv. to prompt}&
        \multirow{2}{*}{\tiny equiv. to prompt}
\\
        &
        &
\\
    \cline{1-1}
        \multirow{2}{*}{prompt}&
        $\leq \nu$&
        $\mathcal O(\log\nu)$
        \\
        &
        (\Cref{thm:expost-det-ub})&
        (\Cref{thm:expost-rand-ub})
        \\
    \hline
\end{tabular}
        \caption{Non-myopic agents\\with \emph{private} graph edges}
        \label{table:private}
    \end{subtable} 
    \caption{\footnotesize Summary of our results, with $\nu=\min(m,n)$, where $n$ is the number of agents and $m$ the number of items. 
    }
    \label{table:results}
    \vspace{-25pt}
\end{table}

\paragraph{Non-myopic agents with public graph edges.} 

Next we consider non-myopic agents who can strategize about their values but not about their desired items: upon the arrival of an item, the set of agents interested in it is revealed (no strategizing involved), but the agent values are reported by the agents themselves. This variant is single-parameter, for which Myerson's lemma applies \cite{Myerson81}.
We prove that, if the mechanism is allowed to wait until the end of the online phase to set prices (i.e., tardy mechanism), then it is possible to achieve the same bounds as in the myopic case, subject to showing that the \greedy and \pgreedy algorithms induce a certain form of {\em global monotonicity}. For prompt mechanisms, in contrast, we establish strong impossibility results. For deterministic mechanisms, we prove a $\nu=\min(m,n)$ competitive lower bound, where $n$ and $m$ denote the number of agents and items, respectively. This sharp deterioration from tardy to prompt mechanisms occurs since in order to prevent buyers from strategizing over future rounds, the prices must be non-decreasing.
Tardy mechanisms circumvent this by assigning payments to agents in the end of the entire process.
A matching $\nu$ upper bound is inherited from the more general case of \emph{private graph edges}, presented below.
For \emph{randomized} prompt mechanisms we establish an $\Omega(\log\nu/\log\log\nu)$ lower bound, using Yao's Minimax principle.
Starting from a carefully designed distribution of problem instances with exponentially increasing agent valuations, we employ a primal-dual approach together with our previous observations on the behavior of \emph{deterministic} truthful mechanisms to bound the achievable competitive ratio.  An almost matching upper bound is inherited from randomized non-myopic prompt mechanisms with private graph edges. The results for non-myopic agents with public graph edges are presented in \Cref{table:public}.\vspace{-7pt}

\paragraph{Non-myopic agents with private graph edges.}   
We finally consider non-myopic agents when both valuations and the set of desired items are private information. 
For deterministic prompt mechanisms, the $\nu$ lower bound from the case of public graph edges applies. 
Moreover, we show that in the case of private edges, every deterministic truthful mechanism is essentially prompt.
Thus, tardy mechanisms for this case retain the $\nu$ lower bound, exhibiting a large gap between tardy mechanisms for public vs. private edges.
We then provide a prompt truthful deterministic mechanism that is $\nu$-competitive, matching the lower bound.
For randomized prompt truthful mechanisms, the $\Omega(\log\nu/\log\log\nu)$ lower bound from the case of public edges applies. This lower bound extends to tardy randomized mechanisms as well, since these are probability distributions over deterministic mechanisms and, as stated above, all deterministic truthful mechanisms for private edges are prompt. 
On the positive side, we provide a randomized prompt truthful mechanism that gives an almost matching competitive ratio of $O(\log \nu)$. 
This algorithm is based on an explore-exploit approach specifically tailored to our case.  

\paragraph{Ex-post vs. ex-ante truthfulness.}
Finally, we explore the notion of \emph{ex-ante} truthfulness, as opposed to \emph{ex-post} truthfulness, where agent's true declarations maximize their expected utility instead of their utility in \emph{any realization} of the random choices of the mechanism. Clearly, ex-post truthfulness implies ex-ante truthfulness.
In the setting with myopic buyers, we only need to consider ex-post truthfulness as we obtain tight approximation in this stronger model that closes the problem also for the ex-ante analogue.
In the setting of non-myopic buyers, we show that the additional hardness introduced by truthfulness cannot be fully attributed to the fact that we require ex-post truthfulness. 
Specifically, we establish a lower bound of $2$ for the competitive ratio of ex-ante truthful mechanisms for this setting (even with respect to randomized tardy ones), exhibiting a gap from the corresponding $e/(e-1)$ upper bound for myopic buyers.  
Our proof utilizes an instance for which we  establish lower bounds on the expected utility of various types of agents. We then employ these to show a contradiction to the mechanism's correctness.

\paragraph{Remark.} 
Throughout the paper, we assume that weights are assigned to vertices (agents) rather than edges. 
Indeed, it is well known that for the more general case of edge weights, even the algorithmic problem is hopeless (see, e.g., Appendix G of \cite{AggarwalGKM11}).
In addition, one may wonder why we do not study the case of non-myopic agents with public valuations but private edges. 
The reason is that in the case of public valuations, it is easy to see that agents cannot benefit from misreporting their edges, implying that \greedy and \pgreedy are inherently truthful.

\subsection{Related Work}
Karp~Vazirani~and~Vazirani~\cite{karp1990optimal} introduced the online matching problem, and studied it under one-sided bipartite arrivals. They observe that the trivial 1/2-competitive greedy algorithm (which matches any arriving vertex to an arbitrary unmatched neighbor, if one exists) is optimal among deterministic
algorithms for this problem. They also provide a groundbreaking and elegant randomized algorithm for this problem, called \ranking, which achieves an optimal $e/(e-1)$ competitive ratio. The work of Karp~Vazirani~and~Vazirani~\cite{karp1990optimal} was extended to vertex weighted settings by Aggarwal~et~al.~\cite{AggarwalGKM11}, who give an optimal $e/(e-1)$-competitive, randomized algorithm using \emph{random perturbations} of weights by appropriate multiplicative factors. The same bound has been re-proven over the years \cite{birnbaum2008line, devanur2013randomized, feige2019tighter, EdenFFS21}. 
Various extensions of one sided online matching and its economic applications (e.g., display ads) have been widely studied over the years,  see e.g. the  excellent survey of Mehta~\cite{mehta2012online} for further reference. Online matching has also been studied under edge and general vertex arrivals, as well as in different stochastic settings (see e.g., \cite{kesselheim2013optimal, korula2009algorithms, ezra2020prophet,GravinTW21,GamlathKS19,gamlath2019online}).

An important generalization of assignment problems in the form of matchings are combinatorial auctions, where buyers can obtain a \emph{subset} of the available items, instead of just one.
Combinatorial auctions with offline strategic buyers and online items has been recently studied by \cite{DPZ21} for submodular and XOS valuations in the case of myopic buyers - considered also in this work - and in the less constrained setting of items that must not be irrevocably assigned at time of arrival. Deng,~Panigrahi~and~Zhang~\cite{DPZ21} show (for myopic buyers) a sharp separation between submodular valuations, which admit a logarithmic competitive ratio, and XOS valuations, for which a polynomial lower bound is proven. In our work, we prove tight constant bounds for myopic buyers in the important special case of a unit-demand matching setting. 

Cole,~Dobzinski~and~Fleischer~\cite{cole2008prompt} formally introduced the notions of \emph{prompt} and \emph{tardy} mechanism, after observing the severe negative aspects of many existing (tardy) methods. They study prompt trutfhul mechanisms for an online problem that is related to ours, but with some restrictions: while agents can be thought as being on the offline side of the graph, their items of interest are restricted to correspond to form an interval over the online steps (which corresponds to the interval buyers are present). Further, agents report their departure time (which can be public/private) once they arrive, and their arrival time is public knowledge. 
Their work is probably closest to ours in spirit, presenting e.g. a logarithmic-competitive, prompt mechanism for the above, less general variant of our problem with private departures.
The notions of tardy and prompt mechanisms have since been adopted in the literature, see e.g. \cite{AzarK11, Xiangzhong15}.
The model of offline agents and online items has been the subject of extensive investigation in economic theory in dynamic mechanism design. Despite this obvious relation to our setting, there are fundamental differences (see for example \cite{BPTZ18, AS2013, Bergemann2019}). 
In dynamic mechanism design, a strategic buyer learns her valuations at time of arrival of each item. Opposed to our setting, priors on agents' valuations for each online item are usually known beforehand. Finally, in our matching setting the agents' valuations can assume only two values, $v_i$ and $0$, and we consider unit demand buyers instead of additive valuation agents as it is customary in dynamic mechanism design. 

\section{Preliminaries}

We are given a bipartite graph $G=(B, I; E)$, where $B$ is a set of $n$ vertices, corresponding to buyers, $I$ is a set of $m$ vertices, corresponding to items, and  $E\subseteq B\times I$ is the set of edges. As aready mentioned, we denote with $\nu$ the smallest between the number of buyers $n$ and the items $m$.
The set of buyers is known beforehand, while the items arrive one by one in some unknown, possibly adversarial order.
Without loss of generality assume that item $j$ arrives at time $j$.
Each buyer $i$ has two pieces of private information: the set of items she is interested in, and her value $v_i$ if she gets at least one of them (the value for other items is 0).
Upon the arrival of a new item, every buyer declares if she is interested in the current item and, if yes, her value.
Let $\bid{i}{j}$ denote the bid of buyer $i$ for item $j$ (with the convention that $\bid{i}{j}=0$ if buyer $i$ is not interested in item $j$).
Without loss of generality, we may assume that buyers cannot change their declared valuation after they have declared it once\footnote{Mechanisms can ``punish'' such behavior by discarding the buyer from further consideration}, i.e. every nonzero bid of the same buyer is the same value $b_i$, and that every buyer is assigned at most one item.

A mechanism $\M$ is composed of an {\em allocation} scheme and a {\em payment} scheme.  Upon the arrival of every item, and based on buyer bids, the mechanism decides immediately and irrevocably to either assign the new item to some buyer who has not been assigned an item yet, or leave it unassigned forever. 
Thus, the resulting allocation is a matching in $G$: every buyer receives at most one item, and every item is allocated to at most one buyer.
We denote by $\mu$ the induced matching, so that $\mu_j$ denotes the buyer to whom item $j$ is assigned (we assume that an item $j$ can only be assigned to a buyer who declares interested in $j$). If $j$ is unassigned, we write $\mu_j = \emptyset$. We also write $\mu^{-1}_i$ to denote the item assigned to buyer $i$, with the convention that $\mu^{-1}_i = \emptyset$ if $i$ is left unassigned.
The \emph{allocation} is computed online; i.e., $\mu_j$ is determined using only the bids on items up to $j$. In addition to the allocation, the mechanism decides how much each buyer should pay. A payment scheme is denoted by $\payment{}$, where $\payment{i}$ denotes the non-negative payment of buyer $i$.  We distinguish between two types of {\em payment} schemes, according to the time at which the mechanism determines the payment. 
A {\em tardy} mechanism is one where the payment vector $\payment{}$ is computed in the end of the process. 
A \emph{prompt} mechanism is one where the payment $\payment{i}$ of every buyer $i$ is determined upon the assignment of buyer $i$ (i.e., upon the arrival of item $\mu^{-1}_i$). 
The mechanism's objective is to maximize the {\em social welfare} of the allocation $\mu$, which is the sum of the buyer values for their assigned items. 
The social welfare is given by $\sw(\mu) = \sum_{i\in B} v_i \cdot \ind{(i,\mu^{-1}_i) \in E}$.
Note that a mechanism can also be randomized, so that its allocation is a distribution over matchings. 
In case of a randomized mechanism, we measure its efficiency by the expected social welfare. 
We say that a mechanism gives an $\alpha$ approximation, or is $\alpha$-competitive (where $\alpha \geq 1$), if its (expected) social welfare is at least an $1/\alpha$ fraction of the welfare of a maximum weight matching. 
That is, $\mu$ is $\alpha$-competitive if $ \opt = \sw(\mu^{\star}) \le \alpha \cdot \E{\sw(\mu)}, $ where $\mu^{\star}$ is the maximum weight matching in $G$.

A \textit{bidding strategy} $\Bid{i}$ for buyer $i$ is a sequence of bids $\bid{i}{j}$ that specifies, every time a new item $j$ arrives, whether to declare interest in it and which value to report. The bid $\Bid{i}$ might depend on the bids of the other agents, the actions of the mechanism, and the knowledge the buyers have on the sequence of items. 
Recall that once an agent declares a positive valuation $\bid{i}{j}=\bid{i}{} > 0$ for some item $j$, she cannot change her value thereafter; namely, all bids for future items $j'$ can take the value of either $\bid{i}{}$ or $0$. 
Let $\mathcal{B}$ denote the profile of buyer bidding strategies, and $\mathcal{B}_{-i}$ denote the profile of all buyer strategies excluding buyer $i$. 
We assume that every buyer has a quasi-linear utility function: $ u_i(\M, \Bid{i},\mathcal{B}_{-i}) = v_i \cdot \ind{(i,\mu_i^{-1}) \in E} - \payment{i}.$

A buyer is called \emph{myopic} if upon the arrival of every item $j$, she cares only about maximizing her utility in that round, without considering its effect on future rounds. I.e., upon the arrival of item $j$, she maximizes the utility function $ u_{i,j}= v_i \cdot \ind{(\mu_j=i,\, (i,j)\in E)}-\payment{i}.$
We consider myopic agents only in the context of prompt mechanisms, where the price $p_i$ is determined immediately.
We study the following ex-post notion of truthfulness: $(i)$ A mechanism for myopic agents is \emph{truthful} if it is always in the best interest of a myopic buyer to declare her value truthfully. $(ii)$ A mechanism for non-myopic agents  is \emph{truthful} if an agent maximizes her utility for every realization of the mechanism by declaring her value truthfully.
Finally, we only consider mechanisms that are {\em ex-post individually rational}, meaning that all agents (myopic or not) have non-negative utility, for every realization of the mechanism.

\section{Truthful Mechanisms for Myopic Buyers}
\label{sec:myopic}
    In this section we study myopic buyers and we show that for this class of agents it is possible to make strategy proof the best (non-truthful) algorithms \citep{karp1990optimal}. In particular, we construct a deterministic prompt mechanism that is guaranteed to achieve at least half of the welfare of the best offline matching, and a randomized prompt mechanism that is (in expectation) $e/(e-1)$-competitive with the best offline matching. 
    
    We start describing our deterministic mechanism \tgreedy, that mimics the classical \greedy algorithm for online weighted matching in a way that is robust to strategic bidding. Every time a new item arrives, \tgreedy runs a second price auction \cite{Vickrey61} to allocate it between the remaining (interested) buyers. Since the buyers are myopic, every time a new item arrives, they behave like if it was the last: clearly there is no point in lying about being interested in an item. Moreover, the truthfulness in each step (as well as the individual rationality) is guaranteed by the well-known properties of the second price auction. Note that the mechanism sets the price for item $i$ immediately, so it is prompt. 
    The analysis of the approximation guarantee is also quite simple: the allocation output by \tgreedy is the same one that the standard \greedy algorithm would have computed. It is well known that \greedy is $2$-competitive with respect to the best offline matching (see, e.g., Appendix B of \cite{AggarwalGKM11}), and that this approximation is tight in the class of deterministic algorithms \cite{karp1990optimal}. We summarize these observations in the following theorem.
    \begin{theorem}
    \label{thm:myopic-greedy}
        The deterministic prompt mechanism \tgreedy is truthful for myopic agents and guarantees a $2$ approximation to the best offline matching. The approximation is tight even for (non-truthful) deterministic algorithms.
    \end{theorem}

We complement this deterministic $2$-competitive, simple mechanism with an optimal, randomized $e/(e-1)-$competitive alternative, \ptgreedy, based on \pgreedy of Aggarwal~et~al.~\cite{AggarwalGKM11}. There, each offline vertex is associated with a random multiplier; then, every time one of the online vertices arrives, it is matched to the free neighbor with largest multiplier-value product. To protect from the strategic behavior of agents, \ptgreedy declares - before the beginning of the online phase - publicly all random multipliers, and then implements Myerson's payment rule \citep{Myerson81} for every round. For a formal description we refer to the pseudocode of \ptgreedy, where we maintain the convention that the $\max$ of an empty set is $0$ and thus if $N(j)$ is empty in line \ref{line:allocation}, then $j$ is discarded and the mechanism passes to the next item. The properties of \ptgreedy are summarized in the following Theorem, whose formal proof is deferred to the Appendix.

\begin{algorithm*}[t]
\begin{algorithmic}[1]
\For{each buyer $i$}
\State Draw $x_i$ uniformly at random from $[0,1]$
\State Let $y_i = 1 - e^{x_i - 1}$
\EndFor
\State Reveal publicly all $x_i$ and $y_i$
\For{item $j$ arriving online}
    \State Receive bids for $j$ and let $N(j)$ be the set of agents interested in $j$
    \State Allocate $j$ to $i^{\star} \in \argmax\{b_i \cdot y_i \mid i \in N(j)\}$
    \Comment{Allocation Rule} \label{line:allocation}
    \State Charge $i^{\star}$ with $p_{i^{\star}} =\max\left\{\frac{y_i}{y_{i^{\star}}} b_i \mid i \in N(j) \setminus \{ i^{\star}\}\right\}$ \Comment{Payment Rule} \label{line:payment}
    \State Discard for further consideration $i^{\star}$
\EndFor
 \caption*{\ptgreedy}
 \label{a:learning-model}
\end{algorithmic}
\end{algorithm*}

\begin{restatable}{theorem}{thmmyopicrank}  
\label{thm:myopic-rank}
        The randomized prompt mechanism \ptgreedy is truthful for myopic agents and achieves (in expectation) a $e/(e-1)$ approximation to the best offline matching. The approximation is tight even for (non-truthful) randomized algorithms.
\end{restatable}

\section{Non-myopic buyers with public graph edges}
\label{sec:public}

We now move our attention to a more demanding notion of truthfulness, where agents are assumed to know, and strategize about, the whole sequence of items arriving. Note that this is a strong information asymmetry between agents and mechanism, as the latter only discovers the items as they are revealed online and has no information on the future. As a first step in this challenging model, in this Section we study the case where agents may only lie on their valuations. Our main focus here is on establishing lower bounds, which will naturally extend to the case where the edges of the graph are private information.

\subsection{Tardy truthful mechanisms}
\label{sec:tardy_truthful}
When the graph edges are public knowledge, we can turn once again to using the algorithmic approaches outlined in the previous Section, i.e. \greedy and \pgreedy. Now that agents cannot strategically withhold or misreport the existence of edges, a tardy truthful mechanism can use the whole graph structure (but of course still not the reported value $\bid{i}{}$) when computing the price charged from any buyer $i$. The prompt, round-wise payment rules from our considerations on myopic buyers, however, do not guarantee non-myopic truthfulness. What remains to prove therefore is that these algorithms can be augmented by a different (tardy) payment rule to be made truthful. This is formally done in two steps in the Appendix: first, it is established that the allocations produced are monotone, and then Myerson's Lemma is employed on the whole algorithm. All in all, we obtain the following Theorem. 

    \begin{restatable}{theorem}{thmpublicgreedy}  
    \label{thm:public-greedy}
        There exists a deterministic, respectively randomized, tardy mechanism that is truthful for non-myopic agents with public graph edges and guarantees a $2$, resp. $e/(e-1)$, approximation to the best offline matching. The approximation is tight even for (non-truthful) deterministic, resp. randomized,
        algorithms.
    \end{restatable}

A last observation: while the allocation computed by the mechanisms we just described are analogue to the ones computed by \tgreedy and \ptgreedy, the payments are different! We are still using Myerson's Lemma, but the critical prices\footnote{The critical price paid by an agent is the smallest bid that would have still resulted in an item being allocated to the agent. See Appendix for a formal definition} are clearly different, as they are computed considering the whole run of the algorithm. To see this, consider the following example. There are two buyers, $b_1$ and $b_2$, and two items $i_1$ and $i_2$. $b_1$ is interested in both the items and has a value of $1$, while $b_2$ only cares about $i_1$, with a value of $0.9$. Assume also for the sake of simplicy that the perturbations $y_1$ and $y_2$ of \pgreedy are both $1$. Both versions of \pgreedy would only allocate $i_1$ to $b_1$, but at two different prices: the mechanism for myopic agents would charge $0.9$, while the tardy one for non-myopic agents would wait the end of the second round and charge $0$. 

\subsection{Prompt deterministic truthful mechanisms}
When mechanisms are required to be prompt, the problem becomes much harder despite the fact that each agent's private information is just a single value.
This is due to the online nature of the problem versus the possibly universal knowledge of the buyers, as outlined in the introduction. We first concentrate on deterministic prompt truthful mechanisms, and prove that the scope of these is indeed quite limited.
\newcommand{\crititemproperty}{critical item property}
\begin{definition}[\crititemproperty]\label{def:crititem}
    We say that a deterministic mechanism satisfies the \emph{\crititemproperty} if and only if for every buyer $i$, there exists some $j\in I$ such that for any reported value $b_i$ of $i$, the mechanism assigns $i$ with item $j$, or none at all.
    Note that $j$ may depend on the edges of the graph, and on the values of other buyers.
\end{definition}

\begin{lemma}\label{lemma:crititem}
    Prompt deterministic truthful mechanisms for the problem with public graph edges satisfy the \crititemproperty.
\end{lemma}
\begin{proof}
    For the sake of contradiction, assume that there is a buyer $i$ who gets item $j_1$ at price $p_1$ if she reports a value $\beta_1$ and gets item $j_2$ at price $p_2$ if she reports a value $\beta_2$. Without loss of generality, let $j_1 < j_2$. By truthfulness, the mechanism must give item $j_1$ to buyer $i$ if she reports a value $\geq p_1$ (as far as the mechanism knows, $i$ might not like items after $j_1$, and she would have incentive to lie and report $\beta_1$ if she is not given $j_1$). Thus, we have $p_2 \le \beta_2 < p_1$, where the first inequality comes from individual rationality. But now, buyer $i$ has incentive to report $\beta_2$, in order to get $j_2$ and pay $p_2$ which is less than $p_1$.
\end{proof}

\begin{theorem}
\label{thm:expost-det-lb}
    Any prompt deterministic truthful mechanism for the problem with public graph edges has competitive ratio of at most~$\nu = \min(m,n)$.
\end{theorem}

\begin{proof}
    Consider an instance with $n$ buyers with value $1$ that are all interested in the first item. If there is a buyer $i$, who will never get item $1$ no matter what she reports, then we change the instance so that $i$ has an arbitrary large value and is only interested in item $1$, in which case $i$ will get nothing and the mechanism does not even approximate the optimal social welfare. Conversely, if there is no such buyer, then the \crititemproperty~ states that no other item can be allocated, which gives an approximation ratio of $\min(m,n)$. 
\end{proof}

\subsection{Prompt randomized truthful mechanisms}
    
    Somewhat surprisingly, the previous section has revealed a very large gap
    between tardy and prompt deterministic mechanisms, when the topology of the graph is public knowledge: while tardy mechanisms can be implemented {\em for free}, i.e., maintaining the efficiency guarantees of (non-strategic) combinatorial algorithms, for prompt mechanisms the story is different. After showing that deterministic mechanisms cannot achieve anything better than $\nu$, we now turn our focus towards impossibility results for randomized mechanisms. We utilize a well-known property of randomized truthful mechanisms, which (by definition) make truthful reports utility-maximizing for any outcome of a mechanism's random decisions, even in hindsight: this implies that they are in fact lotteries over deterministic truthful mechanisms, which in turn satisfy the characterizing properties shown in the previous section.

\begin{theorem}
\label{thm:expost-rand-lb}
     Any prompt randomized truthful mechanism for the problem with public graph edges has competitive ratio of at least $\Omega(\log \nu/\log\log \nu)$.
\end{theorem}
\begin{proof}
    Fix any prompt randomized ex-post truthful mechanism for public graph edges. We are going to argue by Yao's principle \citep{Yao77} that its competitive ratio is at least $\Omega(\log \nu/\log\log \nu)$. This holds due to the upcoming \Cref{lemma:distrib}, which shows that there exists a distribution over instances, such that the optimal solutions have expected welfare $\Omega(\log n)$, and a best-possible deterministic mechanism $\M$, since it satisfies the critical item property, outputs solutions with expected value $\mathcal O(\log\log n)$. 
\end{proof}

\newcommand{\stickfigure}[2]{
    \begin{scope}[scale=.2, shift={(5*#1,5*#2-0.65)}]
        \fill[white] (0,1.3) ellipse (0.3 and 0.2);
        \fill[white] (-.3,0.3) rectangle (.3,1);
        \draw (-.5,0)--(-.3,0)--(-.3,1)--(.3,1)--(.3,0)--(.5,0);
        \draw (-.3,0.3)--(0.3,0.3);
        \draw (-.3,1) -- (-.5,0.8) -- (-.3,0.5);
        \draw (0.3,1) -- (.5,0.8) -- (.3,0.5);
        \fill[white] (0,1.3) ellipse (0.3 and 0.2);
        \draw (0,1.3) ellipse (0.3 and 0.2);
    \end{scope}
}
\begin{figure}[t]
    \centering
    \begin{tikzpicture}[scale=1.5,thick]
    \def\n{9}
    \fill[rounded corners,black!10] (-.2,-1.2) rectangle (2.2,-.3);
    \fill[rounded corners,black!10] (-.2,-2.7) rectangle (2.2,-1.3);
    \fill[rounded corners,black!10] (-.2,-4.7) rectangle (2.2,-2.8);
    \node[rotate=-90] at (2.4,-0.75) {$\overbrace{\hspace{1cm}}$};
    \node[rotate=-90] at (2.4,-2) {$\overbrace{\hspace{1.75cm}}$};
    \node[rotate=-90] at (2.4,-3.75) {$\overbrace{\hspace{2.5cm}}$};
    \node[anchor=west] at (2.5,-.5) {Buyers/items of type $3$:};
    \node[anchor=west] at (2.6,-.75) {expected number};
    \node[anchor=west] at (2.6,-1) {$=n\cdot\beta_3 = 9\cdot 1/7$};
    \node[anchor=west] at (2.5,-1.75) {Buyers/items of type $2$:};
    \node[anchor=west] at (2.6,-2) {expected number};
    \node[anchor=west] at (2.6,-2.25) {$=n\cdot\beta_2 = 9\cdot 2/7$};
    \node[anchor=west] at (2.5,-3.5) {Buyers/items of type $1$:};
    \node[anchor=west] at (2.6,-3.75) {expected number};
    \node[anchor=west] at (2.6,-4) {$=n\cdot\beta_3 = 9\cdot 4/7$};
    \foreach \i in {1,...,\n}
        \foreach \j in {1,...,\i}
            \draw[thin] (0.2,-\i/2) -- (2,-\j/2);
    
    \foreach \i in {1, ..., \n} {
        \fill[white]  (2,-\i/2) circle (0.1);
        \draw  (2,-\i/2) circle (0.1);
        \node at (2,-\i/2) {\tiny $\i$};
    }
    \foreach \i/\j/\t in {1/2/3,2/9/3,
    3/1/2,4/3/2,5/5/2,
    6/4/1,7/6/1,8/7/1,9/8/1}
    {
        \def\a{\pgfmathparse{int(pow(2,3-\t))}\pgfmathresult}
        \stickfigure{0}{-\i/2}
        \node at (0,-\i/2) {\tiny $\j$};
        \node at (-3,-\i/2) {$\j$};
        \node at (-2.5,-\i/2) {$\i$};
        \node at (-2,-\i/2) {$\t$};
        \node at (-1,-\i/2) {$1/\beta_\t = 7/\a$};
    }
    \node at (-3,0) {$i$};
    \node at (-2.5,0) {$\sigma(i)$};
    \node at (-2,0) {$t(i)$};
    \node at (-1,0) {value $v_i$};
    \end{tikzpicture}
    \caption{\footnotesize The instance from \Cref{lemma:distrib} with $k=3$ and $n=9$. Items are ordered (from top to bottom) according to their arrival times, and buyers are ordered (from top to bottom) according to $\sigma$ (sort by decreasing types, breaking ties with indices). Preferences of buyers are given by the edges of the graph.}
    \label{fig:analysistight}
    \vspace{-20pt}
\end{figure}
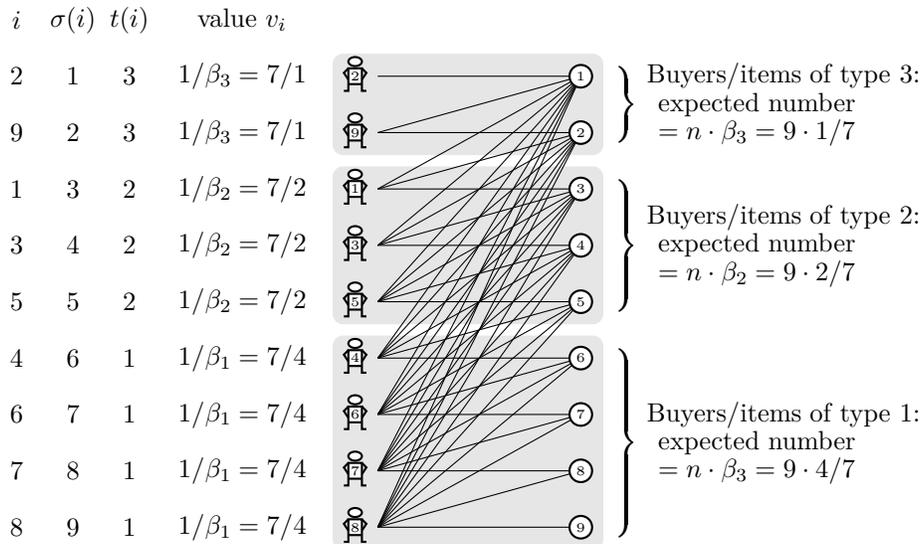

\begin{lemma}\label{lemma:distrib}
    There is a distribution over instances with $n$ buyers and $n$ items, for which optimal solutions have expected value $\Omega(n \log n)$, whereas any deterministic mechanism satisfying the critical item property outputs solutions whose expected value is $\mathcal O(n \log\log n)$.
\end{lemma}

\begin{proof}
Let $k \geq 1$ be a parameter, which corresponds to the number of types of buyers, and let $\beta_1 > \dots > \beta_k > 0$ be the probabilities of each type ($\beta_1+\dots+\beta_k = 1$). Consider the following distribution over instances, with $n$ buyers and $n$ items. 
Each buyer $i$ draws independently a type $t(i) \in \{1, \dots, k\}$ with probability $\beta_{t(i)}$, and we set her value to  $v_i = 1/\beta_{t(i)}$. Then, we sort buyers by decreasing $t(i)$, breaking ties using indices, and call $\sigma(i) \in \{1, \dots, n\}$ the rank of buyer $i$ in this ordering. We decide that buyer $i$ is interested in all items up to the $\sigma(i)$-th one. To visualize this procedure, we refer to \Cref{fig:analysistight}. It is easy to find the optimal allocation: it consists in assigning each buyer of rank $\sigma(i)$ the $\sigma(i)$-th item, in a perfect matching. Thus the expected optimal social welfare is equal to
\[
    \E{\opt} = \sum_{i=1}^n\sum_{t=1}^k \beta_t \cdot 1/\beta_t = n\cdot k.
\]
We now define the type $s(j) = t(\sigma^{-1}(j))$ of an item $j$ as the type of the $j$-th buyer in the ordering~$\sigma$, which corresponds to the type of its buyer in the abovementioned optimal matching. Observe that of each type, there are as many items as buyers, and that buyer $i$ cannot be allocated an item $j$ of type $s(j)<t(i)$.
For each buyer $i$ and for all types $t \leq s$, let $x^i_{s,t}$ be the probability (over the randomness of the types of all buyers except $i$) that $i$ gets an item of type $s$, conditioning on the fact that $i$ has type $t$.
Let $x_{s,t} = \sum_i x^i_{s,t}/n$, that is, the average probability that a type $t$ buyer will be assigned a type $s$ item. The expected social welfare of our deterministic mechanism is equal to
\[
    \E{\sw(\mu)} = \sum_{i=1}^n \sum_{t=1}^k \beta_t \cdot 1/\beta_t \cdot \sum_{s=t}^k x^{i}_{s,t} = n \sum_{t=1}^k \sum_{s=t}^k x_{s,t}.
\]
In expectation, the mechanism sells $\sum_i\sum_{t=1}^s \beta_t \cdot x^{i}_{s,t}$ items of type $s$. Because there are equally many items and buyers of each type, the expected number of items of type $s$ is $\beta_s \cdot n$. Thus, we have the linear constraint 
\[
    \forall 1 \leq s \leq k,\qquad \sum_{t=1}^s \beta_t \cdot x_{s,t} \leq \beta_s.
\]
We are now going to use the \crititemproperty. Fix a buyer $i$, and condition on the types of all buyers except her. We show that there exists an item $\crititem{i} \in \{1, \dots, n\}$, such that for every type $t(i)$, either $i$ gets item $\crititem{i}$, or she gets nothing.
Denote  as $I_t$ the instance given by the fixed types of all buyers except $i$, together with buyer $i$ who has type $t$. Using the \crititemproperty~with instance $I_1$, where $i$ instead is of type $1$ (meaning that $i$ is interested in maximally many items), there is an item $\crititem{i}$ such that buyer $i$ either gets $\crititem{i}$ or nothing. From the perspective of the mechanism, any other instance $I_t$ (defined analogously) is identical to instance $I_1$ up to the point when $i$ stops being interested in items. At this point, if buyer $i$ has already been allocated an item, then it must be $\crititem{i}$. Otherwise, she will not get anything.

Now that $\crititem{i}$ is well-defined (and only depends on types of other buyers), let $y^{i}_{s}$ be the probability (over the randomness of the types of all buyers except $i$) that there exists some type $t$ such that if $t$ is the type of $i$, then item $\crititem{i}$ has type $s$.
Let $y_s = \sum_i y^{i}_s/n$. Because buyer $i$ can only get item $\crititem{i}$, and because $\crititem{i}$ is independent from $t(i)$, we have $x^i_{s,t} \leq y^i_s$. Thus, summing over all buyers, we have the linear constraint $x_{s,t} \leq y_s$, for all $ 1\leq t \leq s \leq k$.
Finally, conditioning on the types of all buyers expect $i$, we show that there is only a small number of types that $\crititem{i}$ can take. Recall that $s(\crititem{i}) = t(\sigma^{-1}(\crititem{i}))$, that is, the type of item $\crititem{i}$ is by definition the type of the $\crititem{i}$-th buyer in the ordering $\sigma$,
where $\sigma$ was obtained by sorting buyers in decreasing order of type.
Consider the ordering induced by $\sigma$ after excluding buyer $i$, and denote $i_1$ and $i_2$ the buyers of rank $\crititem{i}-1$ and $\crititem{i}$. In the original ordering $\sigma$, either $i$ comes before $i_1$ (in which case $s(\crititem{i}) = t(i_1)$), or $i$ comes after $i_2$ (in which case $s(\crititem{i}) = t(i_2)$), or $i$ comes between $i_1$ and $i_2$ (in which case $s(\crititem{i}) = t(i)$). In any case, $t(i_1) \geq s(\crititem{i}) \geq t(i_2)$. This shows that there are at most 2+$z$ possible values for $s(\crititem{i})$, where $z$ denotes the number of types not seen among other buyers.
By a standard computation, the expected value of $z$ is smaller than $\sum_{t=1}^k (1-\beta_t)^{n-1}$. 
Recall that $y_s$ denotes the average probability over $i$ that there exists a type for $i$ which can make $j(i)$ have type $s$, where the randomness is over the instance without $i$. Since for every \emph{fixed} such instance, $j(i)$ can only possibly take two of the types seen in buyers except $i$, for any fixed $i$, it holds 
\[
    \sum_{s=1}^k y_s^i \leq \alpha\qquad\text{where}\quad \alpha = 2+\sum_{t=1}^k (1-\beta_t)^{n-1},
\] 
and therefore, the same holds also on average, i.e. for the $y_s$. 
Thus, averaging over possible types for the other buyers, and summing over $i$, we have the linear constraint $    \sum_{s=1}^k y_s \leq \alpha$. If we choose $n=1+2^k$ and $\beta_t = 2^{-t}/(1-2^{-k})$, we have
\[
\sum_{t=1}^k (1-\beta_t)^{n-1}\leq \sum_{t=1}^k e^{-2^{k-t}/(1-2^{-k})} \leq \sum_{t=0}^{+\infty} e^{-2^t} \leq 1,
\]
and thus $\alpha \leq 3$.
To conclude the proof, we use the linear constraints obtained to define a linear program (P) whose objective function is the expected value of the social welfare obtained by a deterministic truthful mechanism. We want to show that the objective function of our linear program is at most $\mathcal O(n \log k)$. To this end, \Cref{lemma:lp} builds a solution for the dual linear program (D), whose value is an upper bound on the value of the primal linear program (for convenience, the objective function is divided by $n$).
\end{proof}

\begin{figure*}[t]
\centering
\noindent\begin{minipage}{.3\textwidth}
\begin{align*}
  \max& \sum_{t=1}^k\sum_{s=t}^k x_{s,t}\tag{P}\\
  \text{s.t. }& x_{s,t} \leq y_s\\
  &{\textstyle \sum_{t=1}^s} \beta_t \cdot x_{s,t} \leq \beta_s\\
  &{\textstyle \sum_{s=1}^k} y_s \leq \alpha\\
  & x_{s,t},y_s \geq 0
\end{align*}
\end{minipage}
\hspace{25pt}
\begin{minipage}{.3\textwidth}
\begin{align*}
  \min & \ \alpha\cdot w+\sum_{s=1}^k \beta_s\cdot v_s \tag{D}\\
  \text{s.t. } & u_{s,t}+\beta_t \cdot v_s \geq 1\\
  & w\geq {\textstyle \sum_{s=1}^t} u_{s,t}\\
  & u_{s,t}, v_s, w \ge 0\\
\end{align*}
\end{minipage}
\vspace{-15pt}
\end{figure*}

\begin{lemma}\label{lemma:lp}
Consider the linear program (P), parameterized by $\alpha > 0$ and $\beta_1 > \dots > \beta_k > 0$. If $\beta_{t} = 2^{-t}/(1-2^{-k})$ for all $1 \leq t \leq k$, then the dual (B) has a feasible solution of value $\mathcal O(\alpha \log k)$.
\end{lemma}
\begin{proof} Set $\delta = \lceil\log_2 k\rceil$, then following solution of the dual is feasible and yields the desired objective value: $w=\delta$, $v_s = 0$ if $s < \delta$ and $2^{s-\delta}$ otherwise, while the $u_{s,t}$ are defined as:
\begin{align*}
\forall 1 \leq t \leq s \leq k, \quad
u_{s,t} &= \left\{\begin{array}{ll}
    1 &\text{if }s < \delta\\
    1-2^{s-\delta-t} &\text{if }0 \leq  s-\delta \leq t\\
    0 &\text{otherwise}
\end{array}\right.
\end{align*}
\end{proof}

\section{Non-myopic buyers with private graph edges}
\label{sec:private}
Next, we move on to the (harder) case where the graph edges are private information of the agents. The additional hardness, interestingly, severely affects the competitive guarantees only for deterministic truthful mechanisms.
Similarly to before, we begin by characterizing these, and then move on to results for randomized mechanisms.

\subsection{Deterministic truthful mechanisms}

    In the previous section we assumed that the agents could not misreport their interest in items, thus reducing the problem to a single-parameter one. We now lift this assumption, and investigate the effect on the competitive ratio of determistic truthful mechanisms.
    We show that deterministic truthful mechanisms can always be implemented in a prompt manner. Then, we give matching upper and lower bounds on the best approximation ratio for the social welfare.

\begin{lemma}
    \label{lemma:tardy_crititem}
    Tardy deterministic truthful mechanisms for the problem with private graph edges satisfy the \crititemproperty~(see \Cref{def:crititem}).
\end{lemma}
\begin{proof}
    For the sake of contradiction, assume that there is a buyer $i$ who gets item $j_1$ at price $p_1$ if she reports a value $\beta_1$, and gets item $j_2$ at price $p_2$ if she reports a value $\beta_2$. Without loss of generality, we assume that $j_1 < j_2$. First, we argue that $p_1 = p_2$. Indeed, if $p_1 > p_2$ then $i$ with value $\beta_1$ has incentives to lie and report $\beta_2$; whereas if $p_1 < p_2$ then $i$ with value $\beta_2$ has incentives to lie and report $\beta_1$. Second, we slightly change the instance, such that buyer $i$ has value $\beta_2$ and is not interested in items after $j$. When allocating $j$, the mechanism has not seen any difference with the original instance, hence $i$ has incentives to lie and report $\beta_1$ to get $j$, then lie and pretend she was interested in subsequent items to make sure she is charged $p_1$.
\end{proof}

\begin{lemma}
\label{lemma:deterministic_tardy}
    Tardy deterministic truthful mechanisms for the problem with private graph edges are prompt.
\end{lemma}
\begin{proof}
    Assume that our mechanism assigns an item $j$ to a buyer $i$, who reports a value $b_i$. Using \Cref{lemma:tardy_crititem}, the mechanism satisfies the \crititemproperty, and $j$ is the only item which can be assigned to $i$. Let $\pi$ be the minimum value that $i$ could have reported and still be assigned $j$. By truthfulness, $i$ must be charged exactly $\pi$. Indeed, if she is charged $p > \pi$ then $i$ with value $b_i$ has incentives to lie and report $\pi$; whereas if she is charged $p < \pi$ then $i$ with value $p$ would have incentives to lie and report $b_i$. Now, observe that when the mechanism assigns $j$ to $i$, it can retrospectively compute $\pi$, which proves that the mechanism is prompt.
\end{proof}

\begin{theorem}\label{thm:expost-det-ub}
    There exists a deterministic truthful mechanism that achieves an $\nu = \min(m,n)$ approximation of the offline optimum. This result is tight in the class of deterministic truthful mechanisms, when graph edges are private.
\end{theorem}

\begin{proof}
    We start with presenting the positive result. Consider the simple mechanism which only assigns an item to a buyer if she has the highest value seen so far (breaking ties arbitrarily), charging her the second highest value seen so far. It is immediate to verify that this is a deterministic truthful mechanism with an approximation ratio of $\nu = \min(m,n)$.
    For the tightness of the result, \Cref{lemma:deterministic_tardy} shows that deterministic tardy mechanisms are in fact prompt, thus the lower bound from \Cref{thm:expost-rand-lb} (where graph edges are public) applies to this setting.
    
\end{proof}

\subsection{Randomized truthful mechanisms}

Recall that randomized (ex-post) truthful mechanisms are lotteries over deterministic truthful mechanisms, which in turn satisfy the characterizing properties we obtain for the deterministic case. The proof of our lower bound in \Cref{thm:expect-rand-lb} was based on this fact. First, we give a short argument why it also applies to mechanisms for private edges, even when they are tardy. Then, we provide an (almost) matching upper bound, namely a prompt randomized truthful logarithmic approximation.

\begin{corollary}
The $\Omega(\log n/ \log \log n)$ lower bound of \Cref{thm:expost-rand-lb} holds also for the case of private edges, even for tardy mechanisms.
\end{corollary}
\begin{proof}
Fix all random decisions of an ex-post truthful randomized mechanism. This yields a deterministic algorithm, that together with the original mechanism's payment scheme yields a (tardy) mechanism. This mechanism is deterministic, and truthful due to the definition of truthfulness. Also, such a mechanism fulfills the \crititemproperty~(\Cref{lemma:tardy_crititem}), and can even be made prompt (\Cref{lemma:deterministic_tardy}). With this, we can follow the original proof of the lower bound. 
\end{proof}

We state now our prompt mechanism for the problem with private edges, and prove its ratio to almost match our lower bound.
\begin{algorithm*}
\begin{algorithmic}[1]
\Indent{Initialization}
    \State Set  $p \gets 0$ and draw $k \gets \text{Unif}(\{0, 1, \dots, \lceil \log_2 n\rceil\})$
    \State For each buyer $i$, draw type $t_i \gets \text{Unif}(\{Explore, Exploit\})$.
\EndIndent
\Indent{When an item arrives}
    \State Buyers report if they are interested in the item.
    \For{each buyer $i$ of type $t_i = Explore$ who is interested in the item}
        \State Set $p \gets \max(p, v_i/2^k)$
    \EndFor
    \State Sell the item at price $p$ to a buyer $i$ of type $t_i = Exploit$, who is interested\\\hspace{\algorithmicindent}in the item and does not yet has an item, chosen arbitrarily (e.g. lowest index).
\EndIndent
 \caption*{\expexp}
 \label{algo:expost-det-ub}
\end{algorithmic}
\end{algorithm*}

\begin{theorem}
    \label{thm:expost-rand-ub}
    The \expexp is truthful, and computes a $O(\log n)$ approximation to the optimal social welfare.
\end{theorem}

\begin{proof}
    Buyers of type $Explore$ will not get any item, and thus have no incentive to lie.
    Buyers of type $Exploit$ only need to say if they are interested to buy an item at a given price. Because prices are non-decreasing, they have no incentives to misreport their value or their interest in an item.
    For each item $j$, we define $x_j$ as the maximum value seen among buyers interested in items up to $j$.
    \[\forall j \in I,\qquad
    x_j = \max\{v_i \text{ with } i \in B \text{ such that }
    \exists j' \leq j, (i,j')\in E\}\]
    For the sake of analysis, we look at a maximum weight matching $\mu \subseteq E$, having a total value of $OPT$. Each edge $(i,j) \in \mu$ from the optimal solution is assigned to a bucket $\ell_{(i,j)} = \lceil \log_2(x_j/v_i)\rceil \in \mathbb N$. Then for each $\ell \in \mathbb N$ we define $OPT_\ell$ as the total weight of the restriction of the optimal solution to bucket $\ell$.
    \[OPT = \sum_{\ell \geq 0} OPT_\ell
    \qquad
    \text{where }
    \forall \ell \geq 0,\quad
    OPT_\ell = \sum_{(i,j) \in \mu} v_i \cdot \ind{\ell_{(i,j)}= \ell}\]
    Let $V$ be maximum value among buyers who are interested in at least one item. By optimality of $\mu$, the corresponding buyer must be given an item, and thus $OPT_0 \geq V$. Now observe that for each $(i,j) \in \mu$ such that $\ell_{(i,j)} > \lceil \log_2 n\rceil$, we have $v_i < x_j / n \leq V/n \leq OPT_0/n$. Thus, the sum of $OPT_\ell$ for $\ell > \lceil\log_2 n\rceil$ is smaller than $OPT_0$. Therefore, buckets $0, 1, \dots, \lceil\log_2 n\rceil$ contain at least half of $OPT$, that is
    \[\frac{OPT}{2} \leq \sum_{\ell = 0}^{\lceil \log_2 n\rceil} OPT_\ell\]
    For all $\ell \in \{0, 1\dots, \lceil \log_2 n\rceil\}$, we will now show that if $k = \ell$ then the \expexp gives a solution of expected cost at least $\Omega(OPT_\ell)$. Then we will conclude the proof using the law of total probability: summing over $k$ shows that the \expexp computes a solution of expected cost at least $\Omega(OPT/\log n)$. First, assume that $k = 0$. For each edge $(i,j)\in \mu$ in bucket $\ell_{(i,j)} = 0$, then $i$ is the best buyer seen so far. With probability 1/4, buyer $i$ has type $Exploit$ and the second best buyer has type $Explore$. In that case, the \expexp gives buyer $i$  an item (either $j$ or one of the previous items). Using linearity of expectation, the \expexp outputs a solution of expected value at least $OPT_0/4$. Second, assume that $k = \ell$ with $\ell \in \{1, \dots, \lceil\log_2 n\rceil\}$. This case requires an amortized analysis: for each buyer~$i$, denote $X_i$ the random variable equal to $v_i$ if $i$ gets an item and $0$ otherwise; and for each item~$j$, denote $Y_j$ the random variable equal to the value of the buyer to whom $j$ is assigned, and $0$ if $j$ is unassigned. Notice that the \expexp outputs a solution of value $= \sum_{i\in B} X_i = \sum_{j\in I} Y_j$.    Let $(i,j) \in \mu$ be an edge from bucket $\ell_{(i,j)} = \ell$. We are going to show that
    \[\mathbb E[X_i +4Y_j\;|\;k=\ell\text{ and }t_i = Exploit] \geq v_i.\]
    We condition on the fact that $k=\ell$ and $t_i = Exploit$. If buyer $i$ already has an item when item $j$ arrives, then $X_i = v_i$. Otherwise, the best buyer seen so far has type $Explore$ with probability $1/2$, in which case the \expexp gives item $j$ to a buyer of value $\geq x_j/2^\ell \geq v_i/2$.
    Buyer $i$ has type $t_i = Exploit$ with probability $1/2$, thus $v_i \leq  E[2X_i+8Y_j\;|\;k=\ell]$. Summing this last inequality over edges from bucket $\ell$ shows that the \expexp outputs a solution of expected value at least $OPT_\ell/10$.
\end{proof}

\section{Ex-ante truthfulness}
One might wonder if the hardness of truthful mechanisms for our problem is mainly due to the very restrictive notion of ex-post truthfulness. We state here that also for the much looser ex-ante truthfulness, the setting of non-myopic buyers separates clearly from the myopic case. The proof can be found in \Cref{appendix:proof-expect-rand-lb}.

\begin{restatable}{theorem}{thmexpectrandlb}  
\label{thm:expect-rand-lb}
There exists no randomized ex-ante truthful
mechanism that yields an $\alpha$-approximation to the optimal social welfare, for the problem with private edges and any $\alpha<2$. This is true even for tardy mechanisms. 
\end{restatable}

\section{Conclusions}
We have studied vertex-weighted bipartite online matching with offline agents in various settings, obtaining an almost-complete picture of the competitive ratios achievable by mechanisms under different truthfulness notions.
Our results encompass that for myopic truthfulness, the bounds of Karp~Vazirani~and~Vazirani~\cite{karp1990optimal} and Aggarwal~et~al.~\cite{AggarwalGKM11} transfer to the online agents setting. This showcases that the very general myopic bounds of Deng,~Panigrahi~and~Zhang~\cite{DPZ21} are far from tight for restricted settings like ours. On the other hand, we also show that equally near-optimal approximations are impossible under the assumption of classic truthfulness, even ex-ante; and for ex-post truthfulness our seemingly simple problem already exhibits lower bounds almost matching the myopic, logarithmic competitive ratio for submodular combinatorial auctions in Deng,~Panigrahi~and~Zhang~\cite{DPZ21}. We leave open to what extent this additional hardness (moving from a tight $e/(e-1)$ myopic to $\Omega(\log n/\log \log n)$ truthful) already happens when imposing ex-ante truthfulness. This is an interesting subject of investigation, also for different scenarios than the one of our $\geq 2$ lower bound (private edges).
Obtaining according positive or negative results for other variants of online problems with offline agents poses another natural direction.
Besides this, note that our work considers only the (especially hard) case of adversarial arrival order, warranting the question which improved bounds can be obtained e.g. for random-order models. We suspect that non-trivial approximations via (ex-post) truthful mechanisms quickly cease to exist when considering online problems with offline agents that are more general and challenging than ours. On the other hand, under the myopic assumption, these could exhibit interesting bounds situated in between our $e/(e-1)$ and the logarithmic mechanism for submodular combinatorial auctions \cite{DPZ21}.

\section*{Acknowledgment}
Michal Feldman was supported by the European Research Council (ERC) under the European Union's Horizon 2020 research and innovation program (grant agreement No. 866132), by the Israel Science Foundation (grant number 317/17), by an Amazon Research Award, and by the NSF-BSF (grant number 2020788). Federico Fusco, Stefano Leonardi and Rebecca Reiffenh\"auser were supported by the ERC Advanced Grant 788893 AMDROMA “Algorithmic and Mechanism Design Research in Online Markets” and MIUR PRIN grant Algorithms, Games, and Digital Markets (ALGADIMAR).

\bibliographystyle{splncs04}
\bibliography{bibliography}

\appendix    

\section{Proofs of \Cref{thm:myopic-rank,thm:public-greedy}}
    In this section we prove the properties of \ptgreedy and of the tardy versions of \greedy and \pgreedy presented in the main body. Starting from the guarantees of their non-strategic counterparts it is immediate to see that the approximation factor claimed are indeed correct. The only property to show is incentive compatibility. A crucial ingredient to prove incentive compatibility  is Myerson's Lemma, that we recall here for the sake of completeness. The Lemma has been proved in Myerson's seminal paper \cite{Myerson81}; here we follow the more modern approach by Roughgarden~\cite{Roughgarden16}. Since in this paper we study unit-demand agents, we restrict to consider only this type of agents. 

    We start introducing the notion of single-parameter environments. In such environments, there are $n$ agents and a set $X$ of feasible allocations of items to agents. Each agents is characterized by a private valuation to get an item and strives to maximize her quasi-linear utility. To familiarize with this notion consider the model of non-myopic buyers with public graph edges studied in the paper: those agents are indeed single-parameters, as their valuations is their only private information. At the same time, note that the ``edge compatibility'' is implicitly modeled by the following set of feasible allocations of items to agents: an allocation $\mathbf{x} \in \{0,1\}^n$ is feasible if and only if it is corresponds with a matching in the underlying buyers-items bipartite graph. 
    
    As already mentioned in the main body, a mechanism $\M$ is characterized by two features: an allocation $\mathbf{x} \in X$ and a payment rule $\mathbf{p}$. While the allocation specifies who gets what, the payment rule defines how much each agent pays. Allocation and payments are functions of the bids; in particular, we use the notation $x_i(b_i,b_{-i}) \in \{0,1\}$ to specify whether the $i^{th}$ agent is allocated an item, given her bid $b_i$ and the $n-1$ bids $b_{-i}$ of the other agents. We are ready for the following crucial definition: 
    
    \begin{definition}[Monotone allocation]
        An allocation rule $\mathbf{x}$ for a single-parameter environment is monotone if for every bidder $i$ and bids $b_{-i}$ by the other bidders, the allocation $x_i(z, b_{-i})$ to $i$ is nondecreasing in its bid $z$.
    \end{definition}
    
    \begin{definition}[critical prices]
        Fix and agent $i$ and bids $b_{-i}$ of the other agents. Then the critical price for $i$ is defined as the smallest bid $z_i$ such that $i$ is allocated an item, if any. Formally, if we use the convention that the $\inf$ of an empty set is $0$, we have $
            z_i = \inf \{z \, | \, x_i(z,b_{-i}) = 1\}$
    \end{definition}
    
    Clearly, the critical prices enforce ex-post individual rationality. Myerson showed that they also induce (ex-post) truthfulness; we report here a version of Lemma 2 of Myerson \cite{Myerson81} that is tailored to our problem. 

    \begin{theorem}[Myerson's Lemma]
        Fix a single-parameter mechanism. Given any monotone allocation $\mathbf{x}$, it is possible to compute a payment scheme $\mathbf{p}$ such that the resulting mechanism is truthful and individually rational. In particular, in $\mathbf{p}$, each agents that receives an item pays its critical price and $0$ otherwise.
    \end{theorem}
    We are now ready to show the two Theorems. 
    \thmmyopicrank*
    
    \begin{proof}
        We start the proof by arguing that \ptgreedy is truthful and individually rational for myopic agents. First, note that when any item $j$ arrives, there is no point for the buyers still unallocated to lie about their interest for it: if they are not interested and they bid, they would risk to get $j$ and lose future opportunity to get allocated to something they are interested in, while if they are interested they do not want to lose the opportunity (since they have no information on the future, and the prices charged never exceed their valuations). If we restrict to consider the buyers $N(j)$ interested in item $j$, we see that the problem reduces to a single-parameter auction: the agents are myopic and just want to maximize their utility by getting $j$ at a small price. All $y_i$ are public knowledge and non-negative, so our allocation rule (line \ref{line:allocation} of \ptgreedy), fixing these values, is clearly monotone (the more an agent $i$ bids, the more likely she is to exhibit the largest $y_i \cdot b_i$). The allocation is therefore implementable using the Myerson payment rule (line \ref{line:payment} of \ptgreedy). We can conclude, by Myerson's Lemma, that our mechanism is truthful for myopic buyers. Moreover, it is easy to verify that the payment rule enforces individual rationality. Once we have settled the truthfulness, we can assume that all buyers declare their true bids and thus the allocation output by \ptgreedy is the same as \pgreedy~{\em for any realization} of the perturbations $x_i$ and inherits the same approximation: \ptgreedy is $e/(e-1)$-competitive in expectation.
    \end{proof}

    \thmpublicgreedy*
    \begin{proof}
        It is easy to see how the two mechanisms are monotone, thus it is possible to employ directly Myerson's Lemma, as the problem is single-parameter (i.e., the only private information of buyer $i$ is the single value $v_i$). Therefore, \greedy or \pgreedy (with fixed perturbation factors) together with the critical payments defined in Myerson's Lemma result in a truthful mechanism. Note that the greedy algorithm clearly respects our ex-post notion of truthfulness, since no randomization is involved. For the \pgreedy algorithm, this is also true since we fix all random decisions (perturbation) up front, and choose the payment rule accordingly.
    \end{proof}

\section{Proof of \Cref{thm:expect-rand-lb}}
\label{appendix:proof-expect-rand-lb}

\thmexpectrandlb*
\begin{proof}
Fix $\alpha<2$ and assume mechanism $M$ guarantees an expected approximation ratio of $\alpha$.
Consider the following problem instance:
there are $n'$ buyers and $m=n'+1$ items. Every item $j$ has exactly one interested buyer, $i_j$, and all $i_j$ have some small value $v_{i_j}=\epsilon>0$. There exist some additional buyers $B_1\subseteq B$ with different values who are interested only in item $1$, and one buyer, $i$, whom we fix for our considerations. Note that $|B|=n'+n_1$, with $n_1=|B_1|.$ For $n'$ large enough, clearly, $n'\epsilon > \max_{i'\in B_1}v_{i'}$ and the contribution of item $1$ to the optimum becomes negligible with growing $n'$.
Therefore, for $M$ to guarantee an $\alpha$-approximation, there must exist $j\in \{2,\dots,n'+1\}$ such that $i_j$ is assigned the according item with probability at least $\frac{1}{\alpha}$, or in case item $1$ is worth more than $\epsilon$, at least probability $\frac{1}{\alpha}-\Delta_1$, where $\Delta_1$ arbitrarily small for large $n'$.

Now, if we choose $i=i_j$, then $M$ will assign item $j$ to $i_j$ w.pr. $\geq \frac{1}{\alpha}-\Delta_1$, and charge an expected price of at most $\epsilon$. The latter is because the price cannot depend on $i$'s bid due to incentive compatibility, and it needs to be below $i$'s value. 
Assume we replace $i$'s valuation by some $v>\epsilon$, and call this new buyer $i^{(1)}$. Since $M$ is ex-ante truthful, still, the exp. utility $u_{i^{(1)}}$ achieved with a truthful report must be at least as large as when reporting $\epsilon$ instead of $v$, i.e. at least $(v-\epsilon)(\frac{1}{\alpha}-\Delta_1)>\frac 12 v$, which is at least half of $v$ because $\alpha$ is $<2$ and $\epsilon,\, \Delta_1$ can be chosen arbitrarily small. We replace $i^{(1)}$ again by a different buyer $i=i^{(2)}$. She still has valuation $v$, however, she is now interested in items $1$ \emph{and} $j$.
We consider the first step of $M$, i.e. the assignment decision made for item $1$. Assuming that $v$ is the largest value bid on item $1$, and given the fact that $M$ has no idea if any additional value will present itself in the later steps, the probability that $M$ assigns item $1$ to $i^{(2)}$ is at least $\frac{1}{\alpha}-\Delta_2$, where $\Delta_2$ approaches $0$ since the other bids on item $1$ might be, in comparison, too small to matter. Note again that the assignment decision cannot depend on $v$ itself, but only on the fact that it is the largest value bid on item $1$. 

We know that $i^{(2)}$ can get utility larger than $\frac v2$ by simply reporting type $i^{(1)}$ instead.
We also know that since she is assigned item $1$ w.pr. $>\frac 12$, she is assigned item $j$ w.pr. $<\frac 12$. This, intuitively, means that not all of the guaranteed utility is generated by item $j$, not even if the price of $j$ is always $0$ - but some must be generated because her expected price paid when item $1$ is assigned is bounded away from $v$, i.e. $p_{i^{(2)}}(1)= v-\Delta_3$.
In fact, the exp. price $M$ charges from $i^{(2)}$ when assigning item $1$ cannot be smaller if $i^{(2)}$ later reports interest in item $j$, since this would give a buyer of type $i^{(1)}$ incentive to also report interest in $j$. Also, the price charged from $i^{(2)}$ when assigning item $j$ cannot be less than $0$, and when there is no item assigned, $i^{(2)}$ is not charged anything (see preliminaries).
This implies that, for $P_k(i)$ denoting the assignment probability of item $k$ to buyer $i$, 
\[u_{i^{(2)}}=(v-p_{i^{(2)}}(1))\cdot P_1(i^{(2)})+(v-p_{i^{(2)}}(j))\cdot P_j(i^{(2)})= \Delta_3 \cdot P_1(i^{(2)})+(v-p_{i^{(2)}}(j))\cdot P_j(i^{(2)})>\frac v2 .\ \]
Otherwise, we would have a contradiction on the utility being larger than $\frac v2$, i.e. it would be beneficial for $i^{(2)}$ to only report interest in item $j$. In consequence, it also holds 
\[u_{i^{(2)}}=\Delta_3 \cdot P_1(i^{(2)})+(v-p_{i^{(2)}}(j))\cdot P_j(i^{(2)})\geq \Delta_3 \cdot P_1(i^{(2)})+(v-v)\cdot P_j(i^{(2)}) > 0 .\ \]
This is true because the exp. price when receiving item $j$ can be no more than $v$ 
, and $P_j(i^{(2)})<\frac 12$.
Therefore, there exists some $v^-<v$ for which it holds that
\[ u_{i^-}(1)=u_{i^{(2)}}(1) - P_1(i^{(2)})(v-v^-)=(\Delta_3-(v-v^-))P_1(i^{(2)}) \]
Here, $u_{i^-}(1)$ denotes the utility obtained from being assigned item $1$ of some buyer with valuation $v^-$ for item $1$, and $0$ otherwise, when she reports $i^{(2)}$ as her type.
Note that if buyer $i^-$ reports value $v$ for item $1$ and $0$ for all others, she will also obtain $u_{i^-}(1)$ from being assigned the first item: the assignment decision is made before the algorithm can know the difference, and the expected price paid cannot depend on the buyer's later reports due to truthfulness. 

We use this to show a contradiction to the approximation ratio of $M$. 
Assume there exists, in absence of $i^{(2)}$, such a buyer $i^-$ with smaller value $v^-$ and utility of $u^-(1)>0$ when reporting to have value $v$, who is interested in purchasing item $1$, i.e. $i^-\in B_1$.
Since $M$ is ex-ante truthful, a truthful report for her will also result in positive expected utility of at least $u^-(1)$.
As a direct consequence, it holds also that the probability $P_1(i^-)$ for assigning item $1$ to $i^-$ (when she reports truthfully) is lower bounded, in order to achieve above expected utility, as follows: $P_1(i^-)\geq \frac{u_{i^-}(1)}{v^-}$.
Finally, we copy buyer $i^-$ at least $\frac{v^{-}}{u_{i^-}(1)}+1$ times.
If necessary for tie-breaking, we distort their values a bit.
Our conclusions about $i^{(2)}$'s utility hold once $i^{(2)}$ reports the largest value for item $1$, regardless of other values.
This means, if either of our copied $v^-$ should decide to deviate and report to be valued like $i^{(2)}$ instead, they can recover utility $u_{i^-}(1)$.
As a result, each one of the copies, when reporting truthfully, has at least the same utility, and therefore an assignment probability of at least $P_1(i^-)$. This, in sum, results in a probability of more than $1$ for assigning item $1$, i.e., a contradiction.
\end{proof}

\end{document}